\documentclass{sig-alternate}
\usepackage{graphicx,url}  
\usepackage[T1]{fontenc}   
\usepackage[linesnumbered,ruled]{algorithm2e}
\usepackage[normalem]{ulem}
\usepackage{amsmath}
\usepackage{lipsum}
\usepackage{ntheorem}
\usepackage{multirow}

\newtheorem{theorem}{Theorem}[section]
\newtheorem{lemma}[theorem]{Lemma}

\begin{document}

\CopyrightYear{2015}

\title{Optimal Time-dependent Sequenced Route Queries in Road Networks}

\numberofauthors{6} 

\author{
\alignauthor
Camila F. Costa\\
\affaddr{Federal University of Cear\'a, Brazil}\\
\email{camilaferc@lia.ufc.br}
\alignauthor
Mario A. Nascimento\\
\affaddr{University of Alberta, Canada}\\
\email{nascimento@ualberta.ca}
\alignauthor
Jos\'e A. F. Mac\^edo\\
\affaddr{Federal University of Cear\'a, Brazil}\\
\email{jose.macedo@lia.ufc.br}
\and
\alignauthor
Yannis Theodoridis\\
\affaddr{University of Piraeus, Greece}\\
\email{ytheod@unipi.gr}
\alignauthor
Nikos Pelekis\\
\affaddr{University of Piraeus, Greece}\\
\email{npelekis@unipi.gr}
\alignauthor
Javam Machado\\
\affaddr{Federal University of Cear\'a, Brazil}\\
\email{javam@lia.ufc.br}
}


\maketitle%

\begin{abstract}
In this paper we present an algorithm for optimal processing of time-dependent sequenced route queries in road networks, i.e., given a road network where the travel time over an edge is time-dependent and a given ordered list of categories of interest, we find the fastest route between an origin and destination that passes through a sequence of points of interest belonging to each of the specified categories of interest.  For instance, considering a city road network at a given departure time, one can find the fastest route between one's work and his/her home, passing through a bank, a supermarket and a restaurant, in this order. The main contribution of our work is the consideration of the time-dependency of the network, a realistic characteristic of urban road networks, which has not been considered previously when addressing the optimal sequenced route query. Our approach uses the A* search paradigm that is equipped with an admissible heuristic function, thus guaranteed to yield the optimal solution, along with a pruning scheme for further reducing the search space.  In order to compare our proposal we extended a previously proposed solution aimed at non-time dependent sequenced route queries, enabling it to deal with the time-dependency.  Our experiments using real and synthetic data sets have shown our proposed solution to be up to two orders of magnitude faster than the temporally extended previous solution.
\end{abstract}

\section{Introduction}
\label{sec:introduction}
The optimal sequenced route (OSR) query was originally introduced in \cite{Mehdi}.  It aims at finding the optimal route from an origin location passing through a number of points of interest (POIs), each belonging to a specific sequence of categories of interest (COIs). 
This query has several applications within location-based services or car navigation systems, e.g., planning a trip where one leaves from work towards a bank to withdraw money and then to a shopping mall.  Note that ``banks'' and ``shopping malls'' form two COIs, each with many possible POIs. More importantly, although there may be many banks and shopping malls in the city, an OSR query chooses the one bank and the one shopping mall that, in this order, minimize the total cost, such as the travel time, of the trip.
Some variations and solutions to the OSR query have been proposed in \cite{Eisner,Htoo} and \cite{Yutaka} (c.f. Section~\ref{sec:related}).  

However, the ``standard'' OSR query does not take into consideration a single destination that needs to be reached, e.g., after satisfying the COI traversal criterion one may need to go home.  More importantly, on a typical road network the time a user requires to traverse an edge depends on the departure time, i.e., the fact that the road network is a  time-dependent graph.  We, on the other hand, take this characteristic as an intrinsic part of the problem and propose, as our main contribution, an efficient algorithm to solve the new \emph{optimal time-dependent sequence route} (OTDSR) query, defined as follows:
%
\begin{quote}
\emph{The OTDSR($TDG, s, d, t, Q_c$) query takes as input a {\em time-dependent} graph $TDG$  and returns a route for a user departing at time $t$ from $s$ towards $d$, visiting exactly one POI from each COI in the ordered list $Q_c$, such that this route is optimal in terms of total travel time.}
\end{quote}
  
Since time is essential in this case, we also consider the time the user will spend at each POI (such assumption is not considered nor necessary in the standard OSR query).  The case where that such time is itself time-dependent (e.g., being served at a restaurant may take longer at noon than at 11am or 1pm regardless of the user) can be trivially addressed by simply substituting a POI node with two nodes, an ``entry'' and an ``exit'' node, modelling the time spent at the POI as the time cost as the time to travel between such nodes.  That is, the problem remains essentially the same and our proposed solution is readily applicable. The only characteristic that changes is the size of the network, namely the number of POI nodes would double.  
Note that this definition does not prevent one to visit a given COI more than once, nor having the origin and destination as the same points in the TDG.  That is, a path such as: leaving home, going to a restaurant, a movie theatre, a shopping mall, another restaurant and then finally back home is a perfectly legitimate route.

We also make the practically reasonable assumption that the (time-dependent) travel cost of each edge in the TDG satisfies the FIFO property, i.e., an object that starts traversing an edge first has to finish traversing this edge first as well. The general time-dependent shortest path problem in TDGs is NP-hard \cite{OrdaRom1990}, but it has a polynomial time solution in FIFO networks. In the context of our problem, the FIFO property guarantees that there is no improvement in travel time if one ``waits'' at a vertex for the ``best time'' to traverse it, thus in what follows such ``waiting'' is not allowed.

Finally, our proposed solution is generic in the sense that it can be also readily applicable to the standard (i.e., non-time dependent) OSR query as well.   Any OSR query instance can be trivially converted to an OTDSR query instance where the edge costs are constant and the time spent at each POI is null.


In order to exemplify the definitions presented above, consider the graph shown in Figure~\ref{fig:network}, which is a representation of a simple TDG. The curves shown in Figure~\ref{fig:graphics} represent the travel time costs of each edge of the network. Let us consider an OTDSR($TDG, s, d, t, Q_c$) query where $t = 18$ and $Q_c = [C_B, C_R]$.  For the sake of simplicity, but without loss of generality, we assume that the time spent at each COI depends only on the COI itself, i.e., it is equal for all POIs in the same category, e.g., the time spend at any bank in $C_B$ and any restaurant in $C_R$ is 15 min and 60 min respectively.


\begin{figure}[htb]
\centering
\includegraphics[width=0.2\textwidth]{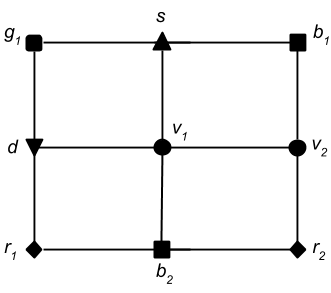}
\caption{A small TDG instance with three different COIs: banks $C_B$ = \{$b_1$ and $b_2$\}, restaurants  $C_R$ = \{$r_1$ and $r_2$\} and a gas station $C_G$ = \{$g_1$\}, a starting point point $s$ and a destination node $d$. Note that $v_1$ and $v_2$ are not POIs, but rather just road intersections..}
\label{fig:network}
\end{figure}

%
%

\begin{figure}[ht]
      \centering
      \includegraphics[width=.2\textwidth]{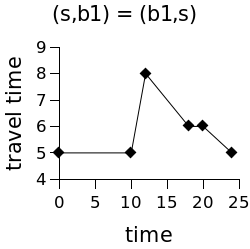}
      \includegraphics[width=.2\textwidth]{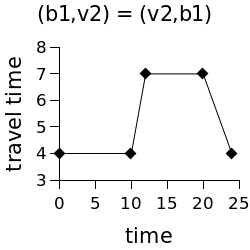}
      \includegraphics[width=.2\textwidth]{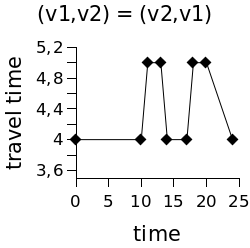}
      \includegraphics[width=.2\textwidth]{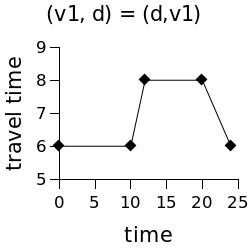}
      \includegraphics[width=.2\textwidth]{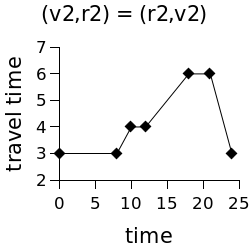}
      \caption{Sample cost functions (travel time in min $\times$ time of day) for the network in Figure~\ref{fig:network}.}
      \label{fig:graphics}
\end{figure}



A greedy approach to solve the standard OSR problem is to perform independent NN searches locally. 
The OSR query cannot be optimally solved by using such an approach \cite{Mehdi,Mehdi2008} and this is also true for the OTDSR query.  
Indeed, if we were to use such a greedy approach for the query instance above we would need to perform a 1-NN query for time-dependent networks, following the approach in \cite{cruz}, twice. First we would find the travel time-wise closest node $v \in C_B$ from $s$ at  $t = 18$, which would be $b_1$. The cost of traversing the edge $\langle s ,b_1 \rangle$ at 18:00 is denoted as $c_{(s, b_1)}($18:00$)$ and is equal to 6 min. Thus, the arrival time at $b_1$ would be 18:06. After staying at the POI for 15 min, the user would leave the node (i.e. bank) at 18:21. The next step would be to locate the closest node $u \in C_R$ from $b_1$ at that time. The node $r_2$ is found through the path $\langle b_1, v_2, r_2 \rangle$ with cost 13 min (via $\langle b1, v2 ,r2 \rangle$, costing 6+7 min) and the arrival time in $r_2$ would be 18:34.  After spending 60 min at the restaurant, the user would then leave $r_2$ at $t = $19:34. Finally the destination $d$ would be reached with cost equal to 19 minutes through $\langle r_2, v_2, v_1, d \rangle$ arriving at 19:53 . Thus, the total travel time of the route $\langle s, b_1, v_2, r_2, v_2, v_1, d \rangle$ is 38 minutes.  However, the optimal one, i.e., the fastest route, is actually $\langle s, v_1, b_2, r_1, d \rangle$ with travel cost equal to 21 minutes. Note that as those are the total travel time, it does not include the time spent at the POIs, since this time is set by the user and will be the same in any solution.  Nonetheless it must be accounted for during the query processing as it will influence the departure time (hence travel cost) from any POI.


The main contribution of this paper is a method to find optimal solutions for OTDSR queries.  The method is based on the A* search algorithm \cite{goldberg} that is equipped with a suitable and provably admissible heuristic function, which guarantees the optimality of the search result. For comparison purposes, we also present a solution that is obtained by extending the also optimal PNE (Progressive Neighbor Exploration) approach proposed in \cite{Mehdi} to cope with the time-dependency characteristic of our problem of interest.  Our extensive experiments using synthetic and real networks show that our algorithm is more efficient than the one using the extended PNE approach.

The remainder of this paper is structured as follows. In Section~\ref{sec:related} we present a brief discussion of related works.
In Section~\ref{sec:proposed_solutions}, we present our proposed approach, and also how we extended the PNE approach. The experimental evaluation and results are shown in Section~\ref{sec:experiments}. Finally, Section~\ref{sec:conclusion} presents a summary of our findings and some suggestions for further work.

\section{Related Work}
\label{sec:related}

The Trip Planning Query (TPQ) was originally proposed in \cite{Feifei}. Differently from the OSR query as well as from the OTDSR query, no order among COIs is imposed in TPQs, which leads to  its NP-Hardness and thus the use of approximation algorithms. 

Optimal solutions to the OSR query in vector and metric spaces were first proposed in \cite{Mehdi} 
and later extended in \cite{Mehdi2008}. 
For vector spaces, \cite{Mehdi} presented a light threshold-based iterative algorithm named LORD, that utilizes various thresholds to filter out the locations that cannot be in the optimal route. The authors also proposed R-LORD, an extension of LORD which uses R-trees to examine the threshold values more efficiently.
For metric spaces and arbritrary size of the COI sequence, they proposed the PNE approach that progressively applies nearest-neighbor queries on different point types to construct the optimal route for the OSR query\footnote{
Since we compare our proposal to a time-dependent extension of the PNE approach, we discuss it in more details in Section~\ref{baseline}.
}.
The main goal of PNE is to replace expensive distance computations of LORD.
In \cite{Mehdi2008} the authors proposed another approach which is applicable for both vector and metric space. They exploited the geometric properties of the solution space and theoretically proved its relation to additively weighted Voronoi diagrams which are recursively accessed to incrementally build the OSR.
That approach is not suitable for comparison to the one we propose because they compute the solution ``backwards,'' i.e., from the last POI to the first one in the desired sequence.  This is not feasible in a time-dependent setting since one would have to know the arrival time at the last POI beforehand, which is itself a function of the departure time that is specified at query time. Thus, we chose to compare our approach to the earlier PNE proposal.

In \cite{Yutaka} the authors proposed a solution that is based on the Incremental Euclidean Restriction (IER) \cite{papadias}. It first finds the shortest OSR given by searches in the Euclidean space, verifies its length in the road network and this value is used as an upper bound. All OSRs that have a length smaller than this value also have the potential to be the shortest route in the road network. Therefore, they must be searched in the
Euclidean space, and then the results must be verified in the road network. The shortest of them is returned as the result. 

In \cite{Htoo} two algorithms that perform an unidirectional and a bidirectional search, respectively, were proposed. Both of them are controlled by an A* algorithm and use the Euclidean distance between a vertex and the destination as the guiding heuristic function. The authors also proposed using a structure called ``visited POI graph'', in order to reduce multiple node expansions.

Two other techniques were presented in \cite{Eisner}. The first, named Iterative Doubling, is an improvement of the EDJ solution \cite{Mehdi}, a Dijkstra-based approach. The OSR is found by performing $l$ Dijkstra runs on the graph, where $l$ is the number of COIs that must be found. Making use of the fact that sequenced route queries tend to be mostly local, this solution is improved by avoiding exploring facilities that are too far. The second solution is based on an extension of the contraction hierarchy \cite{Geisberger}, that was originally proposed as a pre-processing step for ordinary shortest path queries. 

Another query related to the one addressed in this work, named multi-rule partial sequenced route (MRPSR) query, was investigated in \cite{Haiquan}. This query generalizes the OSR and the TPQ queries, by translating ordering constraints into rules. 
When the set of partial sequence rules is empty this query is equivalent to the TPQ query and when the set of partial sequence rules contains one partial sequence rule specifying the  order given by the user, this problem is identical to the OSR problem.

Several other works have studied the problem of finding the fastest route between two locations in time-dependent networks \cite{Xu,Ding,Kim,Kanoulas}. However, these works do not deal with the need to visit specific COIs when traveling between locations.

In \cite{Berube} the authors proposed a travel planning problem which consists of finding the best travel plan from an origin to a destination following a predetermined sequence of POIs in a transportation network with deterministic time-dependent travel times. Other works have considered the problem of finding the best trip plan which visit several predetermined POIs (no order is specified) \cite{Xiang}. Differently from those works, our OTDSR query chooses one (of possibly several) POI from each COI according to a specified COI sequence in order to minimize the route's total travel time.

In \cite{TARS} the authors proposed the traffic-aware route search (TARS) query which aims at finding the fastest route from an origin to a destination via POIs of specified types, while taking into account the time-dependency of the network and the possibility that some visited entities will not satisfy the user. A TARS query may include temporal and order constraints that restrict the order by which entities are visited. Since it is a NP-hard problem, as it is a generalization of the Traveling-Salesman Problem, three heuristics to answer TARS queries were presented. Differently from this approach, we assume that a total order over the COIs is given and, thus, we are able to find an optimal solution.

\section{Solutions to the OTDSR query}
\label{sec:proposed_solutions}

We propose two solutions to solve the OTDSR problem optimally. We first present a baseline solution which is based on the progressive neighbor exploration (PNE) approach in \cite{Mehdi}.  PNE is non-time-dependent thus we extended it in order to cope with the time dependency characteristic of the OTDSR query.  The resulting approach, which we name TD-PNE, uses the TD-NE-A* \cite{cruz} algorithm to perform the necessary time-dependent NN local searches. Next, we propose our solution, called TD-OSR, which uses an A* search to guide the network expansion; we also propose a scheme to reduce the number of node re-expansions.

\subsection{Pre-processing step}
\label{categories}

Aiming at reducing the execution time of the query, a pre-processing step is performed in both solutions, TD-PNE and ours, in order to calculate bounds to guide the expansion of vertices.

Particularly, the TD-PNE solution uses the pre-computed bounds in the TD-NE-A* algorithm \cite{cruz}, which is used to find local time-dependent nearest neighbors. The original TD-NE-A* algorithm aims at finding the $k$ closest POIS from a starting point $s$ at a given departure time $t$ by performing an A* search. 
That algorithm pre-computes lower bounds to reach the closest POI, belonging to any COI, from every vertex in the network. These bounds are used as the heuristic function in the A* search. As we are interested only in POIs belonging to the sequence of COIs given as input, we modified that algorithm so that, in the pre-processing step, we calculate optimistic estimates to reach the closest POI in each COI, from each $v \in V$. Thereby, the search can be guided towards a specific COI. 

In order to calculate these lower bounds, independently of a departure time, we construct a lower bound graph $\underline{G}$. Given a TDG $G$, $\underline{G}$ is a graph that has the same set of vertices $V$ and edges $E$ as $G$ but the cost of an edge $\underline{G}$ is given by the minimum cost possible to traverse the same edge in $G$; $\underline{G}$ is thus an optimistic non-time dependent version of $G$. For example, considering the TDG shown in Figure \ref{fig:network}, the cost of the edge $(s, b_1)$ in $\underline{G}$ would be equal to 5 minutes.

Once $\underline{G}$ is constructed, we pre-compute, from every $v \in V$ and for every $C_i \in C$, the cost $L(v, C_i)$ that is the cost of the fastest path from $v$ to its nearest POI belonging to the category $C_i$ in $\underline{G}$. Note that we do not calculate the cost to reach every POI of the network from every $v \in V$. Moreover, these costs are calculated for every COI because the COI order given as the query input may include any COI and they are not known in advance.


In order to compute these cost estimates, we execute Dijkstra's algorithm from each vertex $v$ to find its nearest POI from each category in $\underline{G}$. In the worst case, this algorithm runs in time $O(\left\vert{E}\right\vert \log \left\vert{V}\right\vert)$ for each execution. As we start a search from each vertex, the total complexity in the worst case is $O(\left\vert{V}\right\vert\left\vert{E}\right\vert\log \left\vert{V}\right\vert)$.  It is noteworthy to emphasize that this is an one time cost which will be amortized over time. Also, the calculated costs do not tend to be updated frequently since they are generally given by the time to traverse the roads when it is possible to drive at the maximum allowed speed.
Regarding space complexity, the space required to store all the information calculated in the pre-processing step is equal to $\left\vert{C}\right\vert \left\vert{V}\right\vert$ bytes, where $\left\vert{C}\right\vert$ is the number of categories of POIs. 

\subsection{PNE's Time-dependent Extension}
\label{baseline}

In order to have a baseline solution, we propose a time-dependent extension of the PNE \cite{Mehdi} algorithm, named TD-PNE. The PNE algorithm generates the optimal route progressively by performing local NN searches in road networks with static costs. 

For the sake of completeness, Figure \ref{fig:pne} illustrates how the original PNE algorithm works. Let us suppose that one departs from the vertex $s$ and wants to visit a bank and a restaurant, in this order. The PNE algorithm first looks for the nearest bank from $s$, $b_1$, and stores this bank and the cost to reach it in the heap (also shown in Figure \ref{fig:pne}). Next, it expands the entry with minimum cost in the heap, ($b_1$: 4), and then it looks for the second nearest bank from $s$ and for the nearest restaurant from $b_1$. It finds $b_2$ with cost 5 and $r_2$ with cost 9. The entries ($b_2$: 5) and ($b_1, r_2$: 13) are inserted in the heap. Note that a complete route $\langle b_1,r_2 \rangle$ was found. The cost of this route is then used as an upper bound, i.e, routes with cost greater than 13 are not investigated. The expansion continues in the same way and the route $\langle s,b_3,r_2 \rangle$ with cost 9 is returned as the optimal one.

\begin{figure}[htb]
   
     \begin{minipage}{.45\textwidth}
     \centering
          \includegraphics[width=0.6\textwidth]{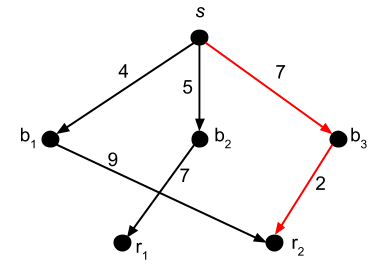}
     \end{minipage}
     \begin{minipage}{.45\textwidth}
         \centering
      \begin{tabular}{| c | l |}
                \hline
                {\bfseries step} & {\bfseries heap contents} \\ \hline
                1 & ($b_1$: 4) \\ \hline
                2 & ($b_2$: 5), ($b_1, r_2$: 13) \\ \hline
                3 & ($b_3$: 7), ($b_2, r_1$: 12) \\ \hline
                4 & ($b_3$, $r_2$: 9) \\ \hline
     \end{tabular}
     \end{minipage}	
   
  \caption{Entries stored in the heap and the POIs found during the execution of PNE algorithm.}
  \label{fig:pne}
\end{figure}

We modified that algorithm such that each NN search accounts for the dependency on departure time and is performed using the TD-NE-A* algorithm \cite{cruz}. The latter is equipped with the new heuristic described above, that takes into account the next COI to be visited in the sequence. For instance, the time-dependent NN query that starts from $s$ is guided towards its closest banks and takes into account the departure time from $s$.

\subsection{The TD-OSR Algorithm}
\label{solution}

We now present TD-OSR, our main contribution in this paper.  It uses an A* search to guide the network expansion so that vertices with more potential are examined first. To measure the potential of a vertex $v$, a function $f(v) = g(v) + h(v)$ is used, where $g(v)$ is the current cost from the starting point $s$ to $v$ , and $h(v)$ is a heuristic function, which in our case is a lower bound estimate of the time to pass through all the categories of POIs and to reach the destination. The smaller the value given by the sum of the current cost to a vertex plus the heuristic function value for it, the greater its potential. If the estimate $h(v)$ is a lower bound to the actual cost ---which is the case as we shall show shortly--- $f(v)$ is said to be admissible and therefore, by virtue of the A* design the solution, i.e., the path found in the TDG, is guaranteed to be optimal.


In order to calculate the $h(.)$ value, we need cost estimates to reach the COIs in the graph as well as to reach the destination $d$. The first is computed in the pre-processing step (described above) and is independent of the sequence given as input. The calculation of the estimate to reach the destination is more difficult because it can potentially be any node of the network. To pre-calculate the shortest path from every node to every other node that could be set as $d$ is an impractical option.  However, once $d$ is given as a query input, we  can compute the shortest path from every vertex $v \in V$ to $d$ by running Dijkstra's algorithm (one-to-all shortest paths) from $d$ in the reverse graph of $\underline{G}$ only once. A reverse graph of $G$ is a graph with the same set of vertices as $G$, but the edges are reversed, i.e., if $G$ contains an edge $(u,v)$ then the reverse of $G$ contains an edge $(v,u)$.  


Once we have the estimates to reach the categories of POIs and to reach the destination computed for any vertex $v$, we can calculate the heuristic function $h(v)$ as follows:

Given the sequence of COIs $C_q = (C_q^1, C_q^2, ..., C_q^m)$ and considering that a path from the starting point $s$ to a vertex $v$ has already passed by POIs of the categories $C_q^1, ..., C_q^i$, the heuristic function value for $v$ is given by: 
\begin{equation}
\label{eq:heuristic}
h(v, C_q[i+1...m], d) = max \{ L(v, C_q^{i+1}),...,L(v, C_q^{m}), L(v, d) \}
\end{equation}
\noindent 
where $L(v, C_q^{i})$ is the cost of the shortest path from $v$ to its nearest POI of the category $C_q^i$ in $\underline{G}$ and $L(v, d)$ is the estimate to reach $d$ from $v$. 
In other words, the cost to pass through the categories of POIs in $C_q = (C_q^1, C_q^2, ..., C_q^m)$ and to reach the destination is at least the minimum cost to reach the farthest of them. 

Let us define two concepts and their notations (which will also be used in the algorithm shortly).
We define as the {\em Travel Time} (TT) of a mobile user the accumulated
cost of traversing all edges in its current path until reaching the
current node $v$, and we denote it  by $TT_v$.  Similarly, we denote
the {\em Arrival Time} (AT) of the mobile user at a node $v$ as $AT_v$.
Moreover, we denote $AT_q = t$ (the query's departure time) and when
traversing an edge between node $u$ and $v$, we have $AT_v = AT_u +
c_{uv}(AT_u)$. 
If we denote $TT(v, d, AT_v, C_q[i+1...m])$ as the remaining travel time leaving from $v$ at time $AT_v$ towards $d$ through POIs in the sequence $C_q[i+1...m]$, it suffices to prove that $h(v, C_q[i+1...m], d) \leq TT(v, d, AT_v, C_q[i+1...m])$, that is: 


\begin{lemma}
\label{lb}
$h(v, C_q[i+1...m], d)$, as defined in Equation~\ref{eq:heuristic}, is a lower bound to the actual optimal remaining cost 
to travel from $v$ to $d$.
\end{lemma}

\begin{proof}
Let us assume  that we have:
\begin{center}
$h(v, C_q[i+1...m], d) > TT(v, d, AT_v, C_q[i+1...m])$ 
\end{center}
or equivalently that:
\begin{center}
$max \{ L(v, C_q^{i+1}),..., L(v, C_q^{m}), L(v, d) \} > TT(v, d, AT_v, C_q[i+1...m])$
\end{center}
We consider two cases: the maximum value is given by $L(v, C_q^{j})$ for $i+1 \leq j \leq m$ or by $L(v, d)$. In the first case we have $L(v, C_q^{j}) > TT(v, d, AT_v, C_q[i+1...m])$, which is not possible because the path from $v$ to $d$ has to include POIs of the categories $C_q[i+1...m]$ in the specified order including a POI of the category $C_q^{j}$ and thus, the cost of this path can not be greater than the lower bound to the cost of reaching a POI belonging to $C_q^{j}$ from $v$. In the second case we have $L(v, d) > TT(v, d, AT_v, C_q[i+1...m])$ which is also not possible because $L(v, d)$ is a lower bound to the cost of reaching $d$ from $v$.  Therefore, the assumption must be incorrect and the proof is complete.
\qed
\end{proof}

Thereby, the heuristic function that guides the A* search in our approach is admissible and therefore the obtained path is guaranteed to be optimal.



Algorithm \ref{alg:TD-OSR} shows the pseudo-code for the TD-OSR algorithm. 
First, the algorithm loads the reverse graph of $\underline{G}$, denoted as $\underline{G}^R$. Then, Dijkstra's algorithm is used to calculate the costs of the fastest paths from $d$ to every $v \in V$ in $\underline{G}^R$. These costs, which are estimates to reach the destination, are stored in vector $destination$. This initialization step is shown in the lines 1 and 2.

\begin{algorithm}[h!]
\caption{The TD-OSR algorithm}
\label{alg:TD-OSR}
 \KwData{A starting point $s \in V$, the departure time $t \in [0, T]$, a sequence $C = [C_q^1,..., C_q^m]$ and a destination $d \in V$}
 \KwResult{The fastest route from $s$ to $d$ considering the departure time $t$ and selecting one POI from each category in $C$ according to the given sequence}
 $\underline{G}^R \gets loadReverse\underline{G}$\;
 $destination[] \gets Dijkstra(\underline{G}^R, d)$\;
 $TT_s \gets 0$\;
 $AT_s \gets t$ \;
 $LB_s \gets h(s, C[1...m], d)$\;
 $R_s \gets []$\;
 Enqueue $(s, AT_s, TT_s, LB_s, R_s)$ in $Q$\;
 \While{$Q \neq \emptyset $}{
 	$(u, AT_u, TT_u, LB_u, R_u)\;  \gets$ Dequeue $Q$\;
 	add $u$ to $Removed[\left\vert{R_u}\right\vert]$\; 
 	\If{$u = d$ and $\left\vert{R_u}\right\vert = \left\vert{C}\right\vert$ }{
 		Return $R_u$\;
 	}
	
 	\For{$v \in adjacency(u)$}{
 		$next \gets \left\vert{R_u}\right\vert + 1$\;
 		$R_v \gets R_u$\;
 		$spent \gets 0$\;
 		\If{$v \in C[next]$}{
 			$spent \gets \tau_{C[next]}$\;
 			$R_v[next] \gets v$\;
 			$next \gets next + 1$\;
 		}
 		$TT_v \gets TT_u + c_{(u, v)}(AT_u)$\;
 		$AT_v \gets (t + TT_v + spent)$ mod $T$\;
 		$LB_v \gets TT_v + h(v, C[next...m], d)$\;
 		\If{v was not removed in a position $\ge$ $\left\vert{R_v}\right\vert$}{
 			\eIf{v is not in Q}{
 			 	 Enqueue $(v, AT_v, TT_v, LB_v, R_v)$ in $Q$\;
 			}{
 				length $\gets \left\vert{R_{v_{in Q}}}\right\vert$\;
 				lb $\gets LB_{v_{in Q}}$\;
 				\If{length $\leq \left\vert{R_v}\right\vert$ and lb $ \ge LB_v$ }{
 					Update $(v, AT_v, TT_v, LB_v, R_v)$ in $Q$\;
 				}
 			}
 		}
 	}
 }
\end{algorithm}

Next, the algorithm begins the network expansion and inserts $s$ in a priority queue $Q$ (line 7) which stores the set of candidates for expansion in the next step. An entry in $Q$ is a tuple $\langle v_i, AT{v_i}, TT{v_i}, LB{v_i}, R{v_i} \rangle$, where $AT{v_i}$ is the arrival time at $v_i$, $TT{v_i}$ is the travel time from $s$ to $v_i$, $LB{v_i}$ is given by $TT{v_i} + h(v_i, C[i...m])$, calculated according to Equation~\eqref{eq:heuristic} and $R{v_i}$ stores the ordered list of POIs that have already been reached in the path from $s$ to $v_i$. The priority of elements in $Q$ is given by increasing order of $LB{v_i}$, thus vertices that offer a greater chance to reach the POIs in $C$ and the destination quickly are checked first.

The vertices are dequeued from $Q$ (line 9) and expanded. When a vertex $u$ is dequeued from $Q$ it is marked as removed (line 10) in the list corresponding to the number of POIs that have already been reached in the path to it. For example, when $s$ is dequeued, it is inserted in the list $Removed[0]$, meaning that no POI, according to the given sequence $C$, was reached in the path to $s$. For every $v$ adjacent to $u$ we check if it is a POI and if it belongs to next category in the sequence. If this condition is satisfied (which is verified in line 18), $v$ is inserted in the next position of $R_v$ and the number of POIs that belong to the sequence found in the path to $v$ is incremented. Furthermore, we add the time the user expects to spend at a POI of the next category, $\tau_{C[next]}$, to the arrival time at the neighbors of $v$ (line 19).

For each $v$, we check whether it is in $Removed[\left\vert{R_v}\right\vert]$ $\ldots$ $Removed[m]$ (line 26). If it is in any of these lists, it is not beneficial for us inserting $v$ in $Q$, since it has been already found in a path that includes more POIs and has a lower $LB_v$. If it is not in any of these queues and neither in $Q$ (which is verified in line 27) it is inserted in $Q$. If it is in $Q$, we check if the new path to it includes more POIs in sequence than in the old one in $Q$ and if $LB_v$ (of the new entry) is less or equal to the one in $Q$ (line 32). If these conditions are satisfied, the old entry of $v$ is removed from $Q$ and the new one is inserted. By doing this, we avoid re-expanding vertices unnecessarily.

The algorithm stops when the next vertex expanded is the destination $d$ and the path from $s$ to $d$, $R_d$, includes all the categories of POIs in $C$ according to the visiting sequence as shown in lines 11 and 12.

\subsubsection{Complexity Analysis}
\label{complexity}

Before analyzing the complexity of our proposed solution, we need to prove the following property regarding the number of times that a vertex is expanded.

\begin{lemma}
\label{lemma:expansion}
A vertex is expanded at most $\left\vert C_q\right\vert + 1$ times in the TD-OSR algorithm.
\end{lemma}
\begin{proof}
Let us assume a vertex $v$ is expanded more than $C_q + 1$ times, i.e., $v$ was inserted in $Q$ at least $C_q + 2$ times. As a path from $s$ to $v$ may include from 0 to $\left\vert C_q\right\vert$ POIs that follows the given sequence,  $v$ has been found through at least two different paths that include the same number of $k$ POIs and the corresponding entries were inserted in $Q$.  However, this cannot happen since $v$ is only inserted in $Q$ if it has not already been found in another path that includes at least $k$ POIs, therefore the assumption is false and the lemma is proven.
\end{proof}

Every time the main loop executes, one vertex is extracted from the queue. Since there are $V$ vertices in the graph and each vertex in inserted in $Q$ at most $C_q + 1$ times, the queue may contain $O(\left\vert C_q\right\vert \left\vert V\right\vert)$ vertices. Each removal operation takes $O(\log\left\vert C_q\right\vert \left\vert V\right\vert)$ time. Thus, the total time required to execute the main loop is $O(\left\vert C_q\right\vert \left\vert V\right\vert \log \left\vert C_q\right\vert$$\left\vert V\right\vert)$. Since each vertex is expanded at most $C_q + 1$ times, at most $O(\left\vert C_q\right\vert \left\vert E\right\vert)$ edges are traversed during the expansion of vertices. For each edge, its ending vertex can be inserted or updated in $Q$, which takes $O(\log\left\vert C_q\right\vert \left\vert V\right\vert)$ time. Therefore, the total run time is $O(\left\vert C_q\right\vert \left\vert V\right\vert \log \left\vert C_q\right\vert$$\left\vert V\right\vert + \left\vert C_q\right\vert \left\vert E\right\vert \log \left\vert C_q\right\vert$$\left\vert V\right\vert) = O(\left\vert C_q\right\vert \left\vert E\right\vert \log \left\vert C_q\right\vert$$\left\vert V\right\vert)$.

\section{Experiments}
\label{sec:experiments}

We performed experiments using both synthetic datasets and a real dataset from the Attica region in Greece. 
We compared TD-OSR, our proposed solution, to the TD-PNE solution according to the (logarithmic) number of expanded vertices in the search and the (logarithmic) processing time of the query.  

\subsection{Synthetic network}
We generated a synthetic time-dependent road network as a regular grid where each grid point corresponds to a vertex in the $TDG$. The POIs were uniformly distributed over the network and each category has, in average, the same number of POIs. To generate the edges cost, for each edge, we chose a random speed between 30 km/h and 80 km/h for each 60 minutes interval of the day, our chosen time granularity, so that the FIFO property is observed within the network.


\begin{table}[htb]
\centering
\begin{tabular}{ c || c }
	{\bfseries Network Size ($\times$1,000)} & 25, {\bfseries 50}, 100 \\ \hline
     {\bfseries Density of POIs} & 0.5\%, {\bfseries 1\%}, 2\%  \\ \hline
     {\bfseries Vertex degree} & 2, {\bfseries 2.5}, 3 \\ \hline
     {\bfseries Number of COIs} & 5, {\bfseries 10}, 20 \\ \hline
     {\bfseries Sequence Size} & 1, {\bfseries 3}, 10 \\ \hline
     {\bfseries Query Locality} & 5\%, {\bfseries 15\%}, 50\% \\
\end{tabular}
\caption{Experimental parameters and their values ({\bf bold} defines default values).}
\label{table:parameters}
\end{table}

We evaluated how each of the solution perform with respect to the parameters shown in Table \ref{table:parameters}. Regarding the network, we varied its size (i.e., number of nodes), the density of POIs (i.e., the number of POIs over the number of vertices), the vertices degree and the number of COIs in the network. Furthermore, we investigated how the algorithms behave in relation to the size of the COI sequence given as input and the distance between the query and the destination vertices, which we refer to as ``Query Locality''. This distance is a percentage of the network diameter and serves to show how harder the problem becomes as this distance increases. For each experiment, we varied a parameter and set the others parameters to their default values. For each combination of different parameters, we generated one TDG and executed 10 randomly generated queries.

\textbf{Effect of the network size.} 
\begin{figure}[htb]
\centering
\begin{tabular}{ c }
	\includegraphics[width=0.3\textwidth]{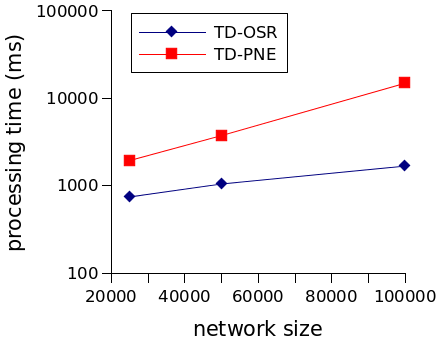} \\  \includegraphics[width=0.3\textwidth]{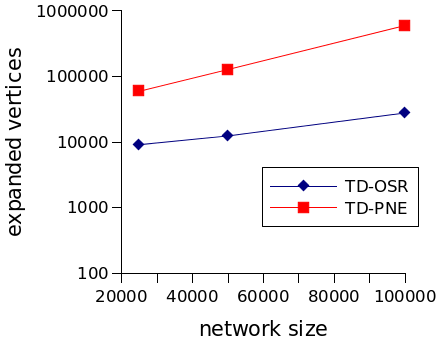}
\end{tabular}
\caption{Processing time of the queries and number of expanded vertices when the network size increases.}
\label{fig:size}
\end{figure}
Figure \ref{fig:size} shows how the solutions behave according to the network size. Both the number of vertices expanded and the processing time increase with the network size in all solutions. This can be explained by the existence of a greater number of paths that can lead to promising POIs in a larger network. It is important to notice that the TD-PNE solution is more affected by this variable than the TD-OSR solution. This happens because the TD-NN searches performed within TD-PNE become more costly with an increase in the number of vertices.
The TD-PNE solution also investigates all the possible routes which have a cost within the cost of the shortest route found so far, not accounting for (an estimate of) the cost of reaching the destination, unlike the TD-OSR algorithm.  
This ultimately leads to investigating routes that are not very promising.


%

\textbf{Effect of the density of POIs.} 
The processing time of the queries and the number of expanded vertices decrease with the density of POIs for the TD-OSR solution whereas it increases slightly for the TD-PNE solution, as shown in Figure \ref{fig:pois}.
When the number of POIs in the network increases, it is easier to find a POI, and thus a route, simply because there are more choices of POIs. This explains the behaviour of the TD-OSR solution. On the other hand, in the TD-PNE solution, even if the TD-NN searches are faster when the POIs become denser, the number of candidate routes to be investigated is greater. With an increasing number of POIs, there is a greater chance of finding partial routes that have a small cost and these routes will be investigated as they offer (misleadingly) a chance to be the optimal route.


\begin{figure}[htb]
\centering
\begin{tabular}{ c }
	\includegraphics[width=0.3\textwidth]{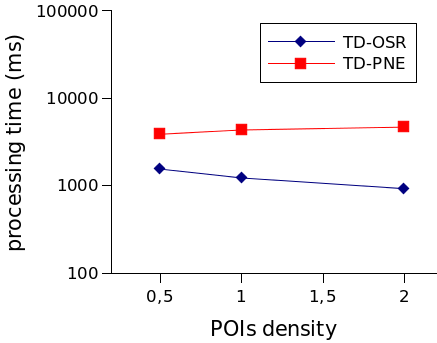} \\ \includegraphics[width=0.3\textwidth]{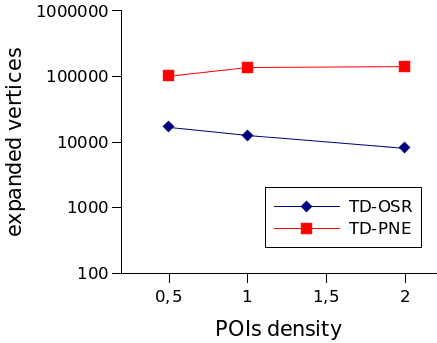}
\end{tabular}
\caption{Processing time of the queries and number of expanded vertices when the POI density increases.}
\label{fig:pois}
\end{figure}

%


\begin{figure}[t!]
\centering
\begin{tabular}{ c c c }
	\includegraphics[width=0.3\textwidth]{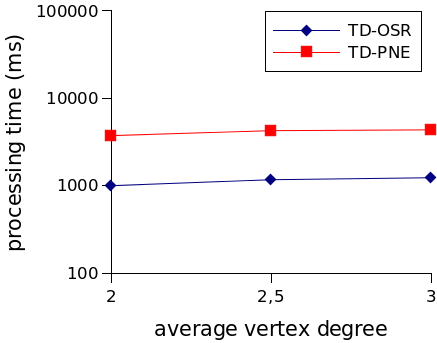} \\ \includegraphics[width=0.3\textwidth]{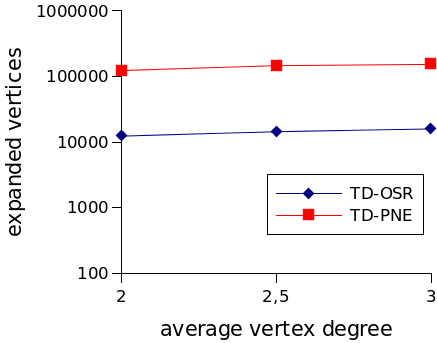}
\end{tabular}
\caption{Processing time of the queries and number of expanded vertices when the degree of the vertices increases.}
\label{fig:degree}
\end{figure}

\textbf{Effect of the degree of the vertices.} As expected and shown in Figure \ref{fig:degree} the processing time of the queries and the number of expanded vertices increase with the degree of the vertices in all solutions. This is reasonable since the number of paths in a network where the vertices have a higher degree is greater, making it more difficult to find POIs and consequently a route.



\begin{figure}[b!]
	\centering
	\begin{tabular}{ c c c }
		\includegraphics[width=0.3\textwidth]{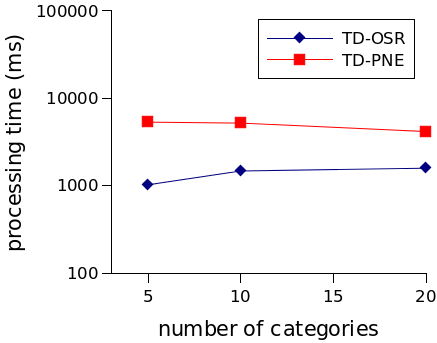} \\ \includegraphics[width=0.3\textwidth]{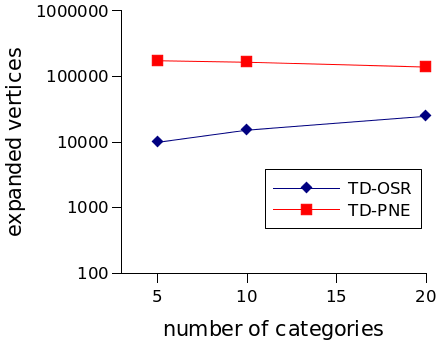}
	\end{tabular}
	\caption{Processing time of the queries and number of expanded vertices when the number of COIs in the network increases.}
	\label{fig:cat}
\end{figure}
\textbf{Effect of the number of COIs in the network.} 
Figure \ref{fig:cat} shows that the processing time of the queries and the number of expanded vertices increase with the number of COIs in the network in the TD-OSR solution while it decreases in the TD-PNE solution. Since the POIs are uniformly distributed among the categories, increasing the number of categories, decreases the number of POIs per COI. Thus, it is harder to find a POI that belongs to a certain category. This explains the behavior of the TD-OSR solution. On the other hand, the TD-PNE solution takes advantage of this because it limits the number of candidate routes to be investigated.



\textbf{Effect of the sequence size.} 
As expected and shown in Figure \ref{fig:seq}, the processing time of the queries and the number of expanded vertices increase with the size of the COI sequence in both solutions. Generally speaking, this is because more vertices have to be checked to find a sequence with more POIs.
Although this variable does not influence the cost of a single TD-NN search, more of those searches need to be performed in the TD-PNE solution when the sequence size increases. Furthermore, in the TD-PNE solution there will be more candidates routes to be investigated. We can also notice that this variable significantly affects the TD-OSR solution.
The greater the number of COIs, the higher the effort of this solution since there are more possible combinations of POIs to form a route that follows the sequence and more of those candidate routes will be investigated.


\begin{figure}[htb]
\centering
\begin{tabular}{ c c c }
	\includegraphics[width=0.3\textwidth]{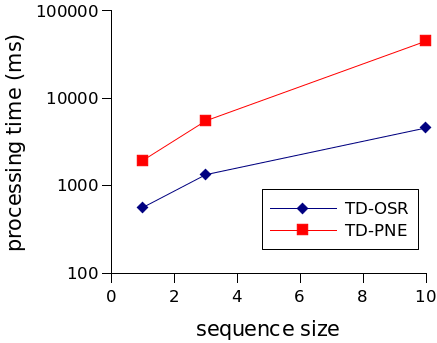} \\ \includegraphics[width=0.3\textwidth]{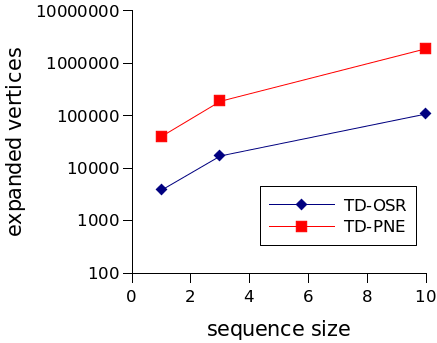}
\end{tabular}
\caption{Processing time of the queries and number of expanded vertices when the sequence increases.}
\label{fig:seq}
\end{figure}

%

\textbf{Effect of the locality of the query.} Figure \ref{fig:loc} shows that the cost of all solutions increase with the distance between the origin and the destination. The larger the distance between the origin and the destination, the greater the number of candidate shortest paths between those two locations. Particularly, more routes will be investigated in the TD-PNE solution because the upper bound found in this solution tends to be high and many partial routes with cost within this upper bound may be found. This involves a larger number of TD-NN searches, which explains why this solution is the most affected by this variable.


\begin{figure}[t!]
\centering
\begin{tabular}{ c c c }
	\includegraphics[width=0.3\textwidth]{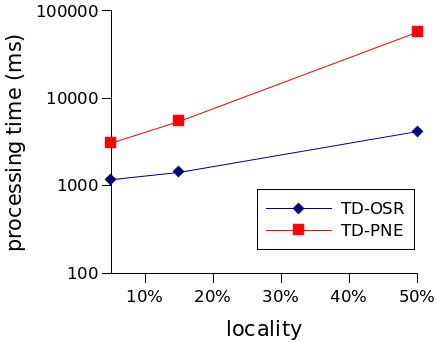} \\ \includegraphics[width=0.3\textwidth]{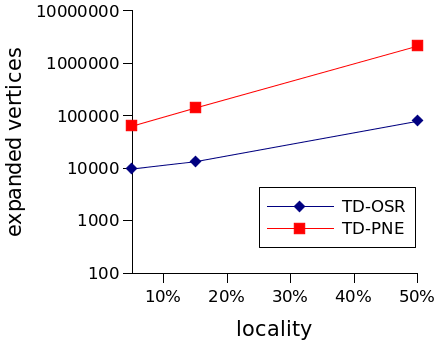}
\end{tabular}
\caption{Processing time of the queries and number of expanded vertices when the distance between the origin and the destination increases.}
\label{fig:loc}
\end{figure}


\subsection{Real network}
The road network of Attica was extracted from OSM\footnote{http://www.openstreetmap.org}. It has 161066 vertices, 204854 edges and a diameter of about 70 km. We chose 20 COIs, that represent services, from the existing ones in the real dataset, making
a total of 1348 POIs, i.e,. an average of 67 POIs per COI. We divided those POIs in ``rare'', ``common'' and ``frequent'' POIs, each such group having 6-7 POIs as a sample of the actual POIs according to their frequency. The most and less frequent POI (gas stations and night clubs, respectively) made for 13\% and 
represented less than 1\% of all POIs, respectively.  
An example of a ``common'' POI is convenience stores, representing 5\% of all used POIs.

To generate the edges costs for the different times of the day, we applied the following procedure. Each edge of the network belongs to a class based on maximum allowable speed. We created one normal distribution for each edge class that returns the speed with which an object is moving on that edge. As the mean for each distribution we used the maximum speed of the class plus the minimum of all maximum speeds of all edge classes divided by 2. The chosen standard deviation was 1/4 of the mean. For each edge, we got a speed value from the corresponding distribution and we checked if it was less than the maximum allowable speed for that edge class. If so, the corresponding travel time is found using the edge length. This time was considered as a base for travels between 12am and 7am and for the time ranges of 8am-9am, 10am-3pm, 4pm-6pm, 7pm-10pm and 11pm-12am) this cost was multiplied by 1.7, 1.4, 1.9, 1.3 and 1.1, respectively, to reflect periods more or less prone to traffic.

Given that this is a real dataset not all parameters can be changed as we did before. In the following we discuss how the algorithms behave in relation to the query locality and the size of the COI sequence. Moreover, we also investigate the effect of the frequency of COIs. 

\begin{figure}[htb!]
\centering
\begin{tabular}{ c c c }
	\includegraphics[width=0.3\textwidth]{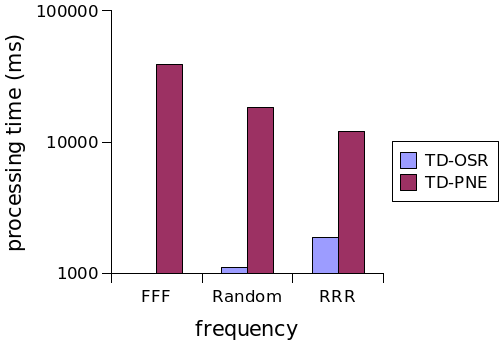} \\ \includegraphics[width=0.3\textwidth]{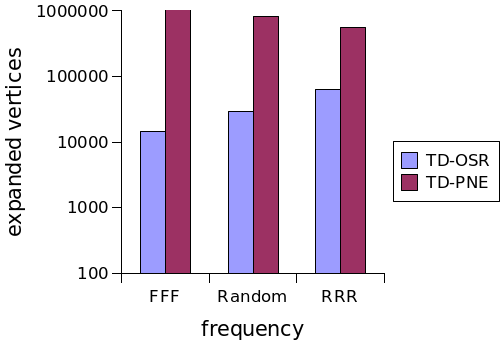}
\end{tabular}
\caption{Processing time of the queries and number of expanded vertices for different sized COIs.}
\label{fig:freq}
\end{figure}

\begin{figure}[htb!]
\centering
\begin{tabular}{ c c c }
	\includegraphics[width=0.3\textwidth]{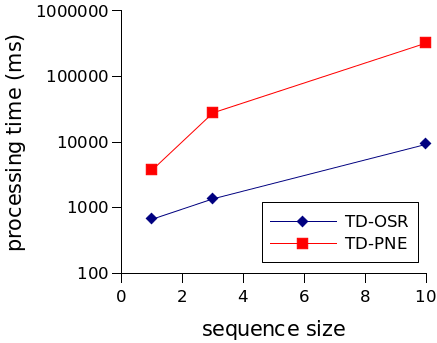} \\ \includegraphics[width=0.3\textwidth]{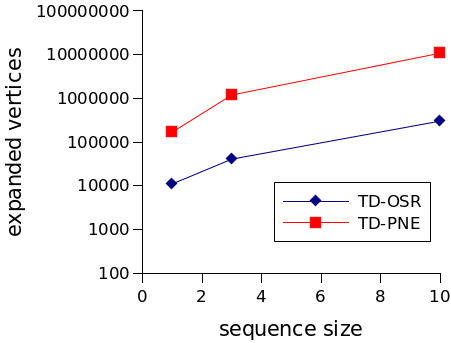}
\end{tabular}
\caption{Processing time of the queries and number of expanded vertices on the real network when the sequence size given as input increases.}
\label{fig:seqReal}
\end{figure}

\textbf{Effect of POI frequency.}
 We generated queries with COI sequences of size 3, but varying the frequencies of those COIs. We investigated the extreme cases of sequences containing only frequent COIs (denoted as ``FFF'') and only rare COIs (denoted as ``RRR'') as well an intermediary case with randomly chosen COIs (denoted as ``Random'').
Figure \ref{fig:freq} shows the results of this experiment. 
Both variables increase in the TD-OSR solution while they decrease in the TD-PNE solution as the POIs query sequence become rare. It takes more effort to find rarer POIs and this explains the behavior of the TD-OSR solution. Nevertheless, in the TD-PNE solution, there are less candidate routes to be investigated as the POIs become rarer, thus the improvement in its performance.  Nonetheless TD-PNE is overall much less efficient than TD-OSR.

\textbf{Effect of the sequence size.} Figure \ref{fig:seqReal} shows that the solutions exhibit the same behavior as in the synthetic experiments performed in relation to this variable. There is a greater number of candidate routes to be investigated in both solutions when the sequence size increases, which explains the decrease in performance as the queried sequence gets larger.



\textbf{Effect of the query locality.} 
Figure \ref{fig:locReal} shows that, as in the synthetic experiments, 
the larger the distance between the origin and the destination, the greater the number of candidate shortest paths between these two locations which increases the number of candidate routes to be investigated in both solutions.

\begin{figure}[htb!]
	\centering
	\begin{tabular}{ c c c }
		\includegraphics[width=0.3\textwidth]{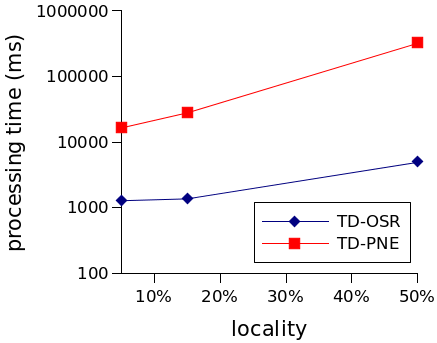} \\ \includegraphics[width=0.3\textwidth]{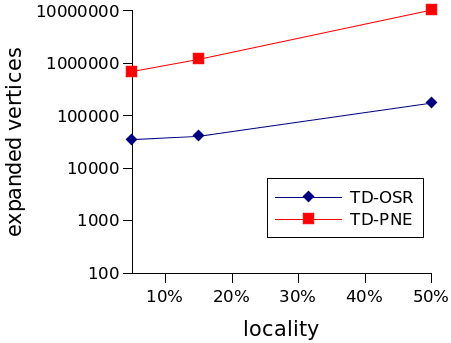}
	\end{tabular}
	\caption{Processing time of the queries and number of expanded vertices on the real network when the query locality increases.}
	\label{fig:locReal}
\end{figure}



\section{Conclusion}
\label{sec:conclusion}
We have addresed an important extension to the known Optimal Sequence Routing (OSR) query, by considering the underlying network to be time-dependent, i.e., the cost (travel time) to traverse an edge depends on the time it is initiated.  This extension, which we call Optimal Time-dependent Sequence Routing (OTDSR) query is a more realistic assumption with respect to, for instance, urban road networks. This formulation of the problem allows one to easily map a traditional OSR query into an OTDSR query, thus making our proposal generic in the sense that it can also address OSR queries.  Finally, our solution is based on the A* paradigm with an admissible heuristic function that guarantees optimality of result.  

In order to evaluate the merits of our proposal, the TD-OSR algorithm, we compared it to a time-aware extension of an approach previously designed for the (non-time dependent) OSR, which we named TD-PNE. 
Varying a number of parameters, and using both synthetic and real datasets the TD-OSR was always significantly faster than the TD-PNE.

As future work we would like to investigate a probabilistic version of the OTDSR problem, i.e., one where the edge costs are not necessarily time-dependent but follow a probabilistic distribution.

\section*{Acknowledgement}

This research has been partially supported by CNPQ grants 306806/2012-6,  Casadinho 552578/2011-8 and 454727/2014-3, European Union project SEEK no 295179; the European Union (European Social Fund - ESF) and Greek national funds through the Operational Program
"Education and Lifelong Learning" of the National Strategic Reference Framework (NSRF) - Research Funding Program: Thales, and  NSERC Canada. 

\bibliographystyle{acm}
\bibliography{sigproc}
\end{document}